\documentclass[10pt]{article}
\usepackage{amssymb}
\usepackage{amsmath}
\usepackage[all]{xy}
\usepackage{amsmath}
\usepackage{amsfonts}
\usepackage{amssymb}
\usepackage{amscd}
\usepackage{amsthm}
\usepackage{latexsym}
\usepackage{amsbsy}
\usepackage{graphicx}
\usepackage{float}
\usepackage[toc,page]{appendix}
\usepackage{caption}
\usepackage{subcaption}
\usepackage{epstopdf}

\newtheorem{theorem}{Theorem}[section]
\newtheorem{definition}{Definition}[section]
\newtheorem{corollary}{Corollary}[theorem]

\newtheorem{prop}{Proposition}
\DeclareMathOperator{\tr}{Tr}

\makeatletter
\def\section{\@startsection{section}{1}{\z@}{-3.25ex plus -1ex minus
		-.2ex}{1.5ex plus .2ex}{\normalfont\bfseries}}
\def\subsection{\@startsection{subsection}{1}{\z@}{-3.25ex plus -1ex
		minus -.2ex}{1.5ex plus .2ex}{\normalfont\itshape}}
\date{}
\makeatother
\title{Bootstrapping Dirac Ensembles}
\author{Hamed Hessam, Masoud Khalkhali, and  Nathan Pagliaroli\\
	Department of Mathematics, University of Western Ontario\\
	London, Ontario, Canada\footnote{\emph{Email addresses}: hhessam@uwo.ca, masoud@uwo.ca, npagliar@uwo.ca}}

\usepackage{graphicx}
	\graphicspath{{./Diagrams/Type(1,0)/}{./Diagrams/Type(2,0)/}}

\begin{document}

	\maketitle
	\begin{abstract}
	We apply   the bootstrap technique  to find the moments of certain multi-trace and multi-matrix  random matrix models suggested by  noncommutative geometry. Using  bootstrapping   we are able to find the relationships between the coupling constant  of these models and their  second moments. Using the Schwinger-Dyson equations,  all other moments can  be expressed  in terms of the coupling constant  and the second moment.   Explicit relations for higher mixed moments are also obtained.  
	\end{abstract}
\section{Introduction}

In this paper we use the bootstrap method to find the moments of certain {\it multi-trace} and {\it multi-matrix} random matrix models inspired by noncommutative geometry. These large $N$ limit moments satisfy an infinite system of nonlinear equations, due to Migdal \cite{Migdal},  called {\it loop equations}. The loop equations are consequences of Schwinger-Dyson equations (SDE's).  The SDE's put constraints on the moments and these constraints help to narrow the search for moments, but this is usually not enough. The process of further narrowing the search space by adding certain extra {\it positivity constraints} is called {\it bootstrapping}. This idea  was recently used by Anderson and  Kruczenski in the context of lattice gauge theory \cite{Anderson}. Then in a  
   random matrix setting  it was used by Lin \cite{bootstraps},   and will be employed throughout this paper as well.  We will see that  further positivity constraints are obtained from the fact that our matrix ensembles originate from  Dirac ensemble. This is an added feature that is absent in standard matrix models.  By narrowing down the search space one can sometimes recover the values of the initial moments. From there the loop equations can be used, in theory, to find any moment. Using the bootstrap technique we are able to find the relationships between the coupling constant of the model and the second moment.
From there all other moments can  be expressed  in terms of the coupling constant  and the second moment, allowing them to be computed explicitly. We also obtain explicit relations for higher mixed moments.

By a {\it Dirac ensemble}   we mean an statistical ensemble of {\it  finite real spectral triples} where the Fermion space is kept fixed, but the Dirac operator is allowed to be random subject to constraints of a  real spectral triple. Technical definitions will be given further below. Such ensembles were first defined by Barrett and Glaser \cite{Barrett2016} with the goal of building  toy models  of Euclidean quantum gravity over a finite  noncommutative space.  They studied these models via computer  simulation, in particular  by Markov chain Monte Carlo  methods. They also  indicated phase transition and multi-cut regimes in their spectral distribution.  Quite significantly  they also noticed, numerically and for  particular models, that at the phase transition point  the limiting  spectral distribution of their  models resembles   the Dirac eigenvalue distribution of a round sphere. That is they have  a manifold  behaviour at phase transition. One expects  all such models to have a manifold type behaviour at phase transition points, but we are still far from proving this attractive conjecture rigorously. We shall indicate the main reason for this further below  in this introduction. 
In \cite{Shahab thesis, AK, Sanchez, Sanchez2, Sanchez3, First paper, Second paper, Gesteau},  formal and analytic aspects of these models and their generalizations  are studied through topological recursion techniques as well as standard random matrix theory methods.

In these models  the integration over metrics is replaced with integration over Dirac operators,
\begin{equation} \label{pf}
Z  =     \int_{metrics}  e^{-S(g)} D (g)  \quad \Rightarrow  \quad \mathcal{Z} = \int_{Diracs} e^{- S (D)}  d D. \, 
\end{equation}
This can be justified by general principles of noncommutative geometry to be explained  below.  In particular Dirac operators are  taken as  dynamical variables and play the role of metric  fields in gravity.   It is a feature of these models that the moduli space of Dirac operators are typically  finite dimensional vector spaces. The action functional $S$ is chosen in such a way that the partition function $\mathcal{Z}$ 
 is absolutely convergent and finite. For example, $ S (D)= \text{Tr} (f (D))$ for a  real polynomial $f$ of even degree with a positive leading coefficient. Note that, in general,  these matrix integrals are not necessarily convergent,  or even need not have a real valued potential. However, they may always be interpreted as formal matrix integrals, which are the generating functions of certain types of maps \cite{Eynard2018}. We mention  that the models considered by Barrett and Glaser and in this paper are always convergent. For more details on formal and convergent matrix models see \cite{Eynard combin}. For a treatment of these integrals when convergent we refer the reader to \cite{First paper}.

The backbone of a spectral triple is the   data $(\mathcal{A},\mathcal{H}, D)$, where $\mathcal{A}$ is an involutive complex algebra acting by bounded operators  on a  Hilbert space $\mathcal{H}$, and  $D$ is a  self-adjoint (in general unbounded) operator acting   on $\mathcal{H}$. This data is required to satisfy some regularity   conditions. For finite spectral triples these conditions are automatically satisfied. A real spectral triple is equipped with two extra operators $J$ and $\gamma$, the charge conjugation  and the chirality operator. 
Finite dimensional real spectral triples have been fully classified by Krajewski in 
    \cite{Finite spectral triples}. Other references include 
	   \cite{QFTNCG, Marcolli, Barrett2015, VS}. In this paper we work exclusively with finite dimensional real spectral triples introduced by Barrett in \cite{Barrett2015}.  Such  finite real spectral triples represent  a noncommutative finite  set equipped with a metric.

Using the classifications of finite spectral triples  and their Dirac operators  one can express the Barrett-Glaser  models as  multi-trace and  multi-matrix  random matrix models 
\cite{Barrett2015, Barrett2016} (cf. also  \cite{Shahab thesis, AK, Sanchez, Sanchez2, Sanchez3}). Since most of the results in random matrix theory is for single matrix and single trace models, this shows that   the analytic study of these models 
as convergent matrix integrals is quite difficult in general.  The scarcity of analytic tools is one of the reasons we are still far from a rigorous proof  of the manifold type behaviour of these models at phase transition. 
The first analytic treatment of these models was carried out in \cite{First paper} where the phase transition for type $(1, 0)$ and $(0, 1)$   models was rigorously proved. In a recent paper the large $N$ limit spectral density function of Gaussian  Dirac ensembles is obtained and shown to be given by   the convolution of the Wigner semicircle law  with itself \cite{Second paper}.   We also mention that in \cite{Sanchez, Sanchez2}  the algebraic  structure of the action functional of these models is further analyzed and linked to free probability theory.  This gives more hope for analytic treatment of these models in general.

One of the motivations behind the introduction of  these models in \cite{Barrett2016} is the  well known observation  that a combination of the Heisenberg uncertainty principle  and Einstein's general relativity will force the spacetime to lose its classical nature as a pseudo-Riemannian manifold at Planck length. This is essentially due to formation of black holes when we probe the space at Planck length. There are many proposals as to what should replace the classical spacetime.  A noncommutative space in the sense of spectral triples
is an attractive proposal since the metric, as a necessary dynamical variable of a theory of gravity, is already encoded by  the Dirac operator. Furthermore,  assuming the spectral triple to be finite,  allows computer simulation as well as methods of random matrix theory to be applied and various scenarios to be tested. There is also a more ambitious hope of  coupling  these models with fermions, as they appear in the finite spectral triple $F$ of the standard model of elementary particles \cite{QFTNCG, VS}.  This will be  along the lines of the work of  Chamseddine, Connes and Marcolli \cite{Neutrino Mixing} where they consider spectral triples of the form $X \times F$, where $X$ is a Riemannian manifold  which  represents the gravitational sector.   Using the spectral action principle \cite{Spectral action} and heat kernel expansion they were able to obtain the standard model Lagrangian coupled with gravity. For a recent account and survey see \cite{Advances in NCG}. Again instead of a manifold, one needs to use a noncommutative space and a finite real spectral triple is a good first approximation to that.  We should mention that the initial steps in this direction has already been taken in \cite{Gesteau}, and also in \cite{Sanchez3}. Specially in the latter work the kinematics of coupling  the gravitational field, in the context of finite spectral triples, with the Yang-Mills-Higgs  field of the standard model is worked out. What remains to be done, and that is a tall order indeed,  is to choose a suitable potential $S$  to go in the path integral 
\eqref{pf}, and to  calculate various quantities of interest and also study the large $N$ limits of them. We hope to come back to this project  in the  near future.

The idea of  replacing metrics by Dirac operators  is justified by general principles of noncommutative geometry as we explain next.
For example   the distance formula of Connes  \cite{Connes94}
	$$ d(p, q) = \text{Sup}\{|f(p)-f(q)|; \, ||[D, f]|| \leq 1\},$$
	 shows that the geodesic distance on a Riemannian spin manifold  can be 	recovered from the action of its Dirac operator $D$ on the Hilbert space of spinors. This role of the Euclidean Dirac operator as a selfadjoint elliptic operator can be abstracted and cast in  the notion  of  a (real) spectral triple which  simultaneously encodes  the data of a  Riemannian metric,  a spin structure, and a  Dirac operator on  a commutative or noncommutative space   \cite{Connes95}.
	   A  deep result,   the   {\it reconstruction theorem} of Connes  \cite{Connes-2013},   shows that   a spin Riemannian manifold can be fully recovered from  a commutative real spectral triple satisfying some natural conditions. For this reason, real   spectral triples can be regarded as  noncommutative spin Riemannian manifolds, and the space of their compatible Dirac operators as the space of Riemannian metrics. 
	   
Finite real spectral triples  of interest in this paper are characterized by a pair of non-negative  integers $(p, q)$ \cite{Barrett2015,Barrett2016}. These integers count the number of gamma matrices that square to one and minus one, respectively. As $p$ and $q$ increase, the corresponding Dirac ensemble gets more and more complicated as a multi-matrix and multi-trace matrix model. The signature of the model $s = p+q$ determines the multiplicity of the matrix model. Consider for example the case when $p =1$ and $q = 0$. The Dirac operator can then be expressed as $D = H \otimes I + I \otimes H$ and the partition function  (\ref{pf}) reduces to  
 \begin{equation*}
 	\mathcal{Z} = \int_{\mathcal{H}_{N}}e^{-\tilde{S}(H)}dH,
 \end{equation*}
where $\tilde{S}$ is some new potential function  in $H$ and  $dH$ is the Lebesgue measure on  the space  $\mathcal{H}_{N}$ of $N \times N$ Hermitian matrices. Note that even in this case $\tilde{S}$ is a multi-trace function. We shall see more  examples  later in this paper. 
In this paper we are mainly concerned with finding the moments of these models in the large $N$ limit i.e. when the matrix size approaches infinity. In the above mentioned $(1,0)$ ensemble they are defined as  
\begin{equation*}
m_{k}=   \lim_{N  \rightarrow \infty}\langle \frac{1}{N} \tr H^{k}\rangle = \lim_{N \rightarrow \infty}\frac{1}{N}\,\frac{1}{Z}\int_{\mathcal{H}_{N}}\tr H^{k} e^{-\tilde{S}(H)}dH.
\end{equation*}
Mixed moments are also defined in a similar manner.

Here is a brief outline of the contents of this paper. In Section two we derive  and discuss the Schwinger-Dyson equations, the loop equations,   as well as the crucial idea of mixed moment factorization  in the large $N$ limit. We will then explain the positivity constraints on moments in both single and multi-matrix models. In Section three we will compare the numerical results from bootstrapping to the analytic solution for signature one models obtained in \cite{First paper}. In Section four we will compare features that we found using the bootstrap method   for signature two models with those obtained in \cite{Barrett2016} via  Monte Carlo simulation. In the Appendix we briefly explain how the factorization of mixed moments is obtained from their genus expansion.

We would like to warmly thank the referees for several very useful comments and suggestions that we believe led to 
a better exposition of this paper.

\section{The Bootstrap Method}\label{boot}

\subsection{The loop equations} In general,  the partition function of a   Dirac ensemble is  a  multi-trace and multi-matrix model of the form
\begin{equation*}
\mathcal{Z}:=\int_{\mathcal{H}_{N}^{m}} e^{-\tilde{S}( H_{1},  H_{2},..., H_{m})} dH_{1}... dH_{m},
\end{equation*}
where $\tilde{S}$ is the trace of a  polynomial in traces of the $m$ variables $H_1, \dots, H_m$  and their products with suitable powers of $N$ in the coefficients. This integral is over the Cartesian product of $m$ copies of the  space $\mathcal{H}_N$ of Hermitian $N\times N$ matrices and the integration is with respect to the Lebesgue measure in each matrix variable. Such an integral can be considered either as a formal or convergent matrix integral. Note that both types of models satisfy the same SDE's. The SDE's relate the moments of the model in  some word $W$ in the alphabet of matrix variables $\{ H_{1},  H_{2},..., H_{m}\}$,  defined as  expectation values

\begin{equation*}
\langle \frac{1}{N} \tr W \rangle :=\frac{1}{N}\,\frac{1}{Z}\int_{\mathcal{H}_{N}^{m}}\tr W e^{-\tilde{S}( H_{1}, H_{2},...,H_{m})} dH_{1}...dH_{m}.
\end{equation*}

The SDE's are a common technique used in Random Matrix Theory \cite{Eynard2018,Random Matrices} and can be derived  in the following manner. We shall be very brief. Take a  word $W$ as before and consider the following relation 

	\begin{equation*}
	\sum_{i,j=1}^{N}\int_{\mathcal{H}_{N}^{m}} \frac{\partial}{\partial (H_{q})_{ij}} \left( W_{ij} \,e^{-\tilde{S}( H_{1},  H_{2},..., H_{m})} \right)dH_{1}...dH_{m} =0,
	\end{equation*}
where $W_{ij}$ denotes the $(i,j)$-entry of the product of matrices that make up the word $W$. This relation easily follows from  the Stokes' theorem. The use of the product rule in the left hand side results in the Schwinger-Dyson equations. For example,  when $m =1$, $W = H_{1}^{\ell} $, and $\tilde{S}(H_{1}) =  \frac{N}{2}\tr H^{2}_{1}$,  the above equation generates the following relations
\begin{equation*}
	\sum_{k=0}^{\ell -1} \langle \tr H_{1}^{\ell-1-k} \tr H_{1}^{k} \rangle = \langle N \tr H_{1}^{\ell} \tilde{S}'(H_{1})\rangle =\langle N \tr  H_{1}^{\ell+1}\rangle  .
\end{equation*}

For a finite  $N$ the SDE's put some constraints on moments and in general do not determine the moments. However, if the large $N$ limits of moments exist, then the SDE's  simplify 
 dramatically\footnote{The large $N$ limit is understood differently for formal and convergent models, we again refer the reader to \cite{Eynard combin}. }.  This is a consequence of the factorization property of the large  $N$ limits of mixed moments.  In particular,  when $m = 1$, it is well known that the following factorization holds in the large $N$ limit
\begin{equation*}
 \langle \tr H_{1}^{a} \tr H_{1}^{b} \rangle =  \langle \tr H_{1}^{a}\rangle \langle \tr H_{1}^{b} \rangle,
\end{equation*}
for any positive integers $a$ and $b$. This is true for  single trace single matrix convergent models  \cite{Universality},  and for formal multi-trace single matrix models.
This factorization holds also in some formal multi-matrix models ($m >1$),  in particular when those models have a genus expansion. See the appendix for a detailed explanation.
One key assumption we will make for the signature two models is that this factorization does hold.  This is proved in \cite{Second paper}.

One key fact to notice is that the loop equations are relations for generating higher order moments from lower order ones. In particular one might wonder what is the minimum number of generating moments one needs to generate all of the rest. Using the terminology from \cite{bootstraps} we will refer to this collection of moments as the \textit{search space}. What we will see in the following sections is that the search space of the matrix ensembles studied here are of only one dimension.

\subsection{Bootstrapping models with positivity}
The positivity constraints on  moments of a  single matrix model in the large $N$ limit are related  to the {\it Hamburger moment problem}, as it was  noticed by  Lin \cite{bootstraps}. However,  positivity and bootstrapping was already  used in the context of solving the loop equations for lattice gauge theory by  Anderson and   Kruczenski \cite{Anderson}.  
Bootstrapping  has recently been used in several other works on matrix models including \cite{Kazakov, Monte Carlo}, and for matrix quantum mechanics \cite{Matrix QM, numerical bootstraps,  Berenstein}.

The Hamburger moment problem can be formulated as follows: given a sequence $(m_0, m_1, m_2, \dots)$ of real numbers,  one asks if there is a positive Borel measure $\mu$ on the real line so that  $m_k$ is the $k$th moment of $\mu$, that is 

$$ m_k =\int_{\mathbb{R}} x^k  d \mu (x), \quad k=0, 1, 2, \dots. $$
It is known that a necessary and sufficient condition for the existence of $\mu$ is that 
 the {\it Hankel matrix} of the moments 
\begin{align*}
\mathcal{M}=
\begin{bmatrix}
		m_0   &  m_1 &  m_2 & m_3 & \cdots \\ 
		m_1 &  m_2 &  m_3 & m_4 & \cdots \\
		m_2 &  m_3 &  m_4 & m_5 & \cdots \\
		m_3 &  m_4 &  m_5 & m_6 & \cdots \\
		\vdots & \vdots &  \vdots & \vdots & \ddots
\end{bmatrix},
\end{align*}
is a positive  matrix. That  is 
$$ \sum _{j,k\geq 0}m_{j+k}c_{j}{\overline {c_{k}}}\geq 0, $$
for all sequences $(c_j)$ of complex numbers with finite support.  In fact checking the necessity of this condition is quite easy. Take the polynomial  $f(x) = \sum c_j x^j$. Then the  positivity of the integral 
$\int_{\mathbb{R}} f(x) \overline{f (x)}  dx$ immediately implies the  positivity of the Hankel matrix. For a proof of  the sufficiency of the condition  see page 145 of \cite{RS}.

In our context $m_0 = 1,$ 
$$ m_k= \lim_{N\to \infty} \frac{1}{N}\langle \text{Tr} H^k\rangle, \quad  k = 1, 2, \dots,$$
and $d \mu (x) =\rho (x) dx$, where $ \rho(x)$ is the limiting spectral density function.

These positivity constraints can be applied to both the Hankel matrix of moments of the matrix ensemble and the Hankel matrix of moments  of the  Dirac ensemble defined as 

\begin{align*}
\mathcal{D}=
\begin{bmatrix}
1   &  d_1 &  d_2 & d_3 & \cdots \\ 
d_1 &  d_2 &  d_3 & d_4 & \cdots \\
d_2 &  d_3 &  d_4 & d_5 & \cdots \\
d_3 &  d_4 &  d_5 & d_6 & \cdots \\
\vdots & \vdots &  \vdots & \vdots & \ddots
\end{bmatrix}
\end{align*}
where 

\begin{equation*}
	d_{\ell} = 	\lim_{N \rightarrow \infty}\langle \frac{1}{N^2}\tr D^{\ell}\rangle=\lim_{N \rightarrow \infty} \frac{1}{N^2}\, \frac{1}{Z}\int_{\mathcal{G}}\tr D^{\ell} e^{- S(D)} dD.
\end{equation*}
These additional constraints are one advantage that Dirac ensembles have over matrix ensembles when bootstrapping.

There is a generalization of the (univariate) Hamburger moment problem to the non-commutative, multivariate case \cite{Burgdorf}. To discuss this generalization that we need in this paper we require the following definition. 
\begin{definition}
	The sequence of the real number $\{m_w\}_{w\in W}$ indexed by word $w$, where $W$ is the space of the words formed by Hermitian matrices $H_1 ,H_2,\cdots,H_n$, is called a tracial sequence, if $m_w = m_u$ whenever $w$ and $u$ are cyclic equivalent.
\end{definition}
The truncated moment problem asks for necessary and sufficient conditions for a tracial sequence to be a sequence of the moments of some non-commutative, multivariate distribution.

Consider the large $N$ limit of a multi-trace multi-matrix model of the following form
\begin{align*}
	\mathcal{Z}:= \int_{\mathcal{H}_N^m} e^{-S(H_1,H_2,\dots,H_m)} dH_1 \dots dH_m.
\end{align*}

The (infinite) tracial moment matrix $\mathcal{M}(m)$ of a tracial sequence $m = \{m_w\}$ indexed by words is defined by the symmetric matrix
\begin{align*}
\mathcal{M} (m)=(m_{w*u})_{w,u}.
\end{align*} 

The necessary, but not sufficient, condition for a sequence of  $\{m_w\}_{w\in W}$  is positive semi-definiteness of the tracial moment matrix \cite{Burgdorf}.

For instance, with tracial sequence $m_{\emptyset}=1 , m_A , m_B , m_{AA} , m_{AB}  , m_{BB}, \dots,$ we can enforce positivity of the sub-matrix of $\mathcal{M}$, defined as 
\begin{align*}
\begin{bmatrix}
1   &  m_A &   m_B & m_{AA} & m_{AB} & m_{BB} \\ 
m_A & m_{AA} &  m_{AB} & m_{AAA} & m_{AAB} & m_{ABB} \\
m_B &  m_{BA} &   m_{BB} & m_{BAA} & m_{BAB} & m_{BBB} \\
m_{AA} & m_{AAA} & m_{AAB} & m_{AAAA} & m_{AAAB} & m_{AABB}\\
m_{AB} & m_{BAA} & m_{BAB} & m_{BAAA} & m_{BAAB} & m_{BABB}\\
m_{BB} & m_{BBA} & m_{BBB} & m_{BBAA} & m_{BBAB} & m_{BBBB}\\
\end{bmatrix}.
\end{align*}

\subsection{The algorithm}

 A {\it Python} script is first used to generate the loop equations for all possible words up to a given order. This order is dependent on how many loop equations it takes to deduce the dimension of the search space. Once we have found the search space and generated a sufficient number of loop equations we compute all moments, assemble them into the matrices outlined above, then check for positivity of the matrix and various submatrices. This process is done for both matrix moments and Dirac moments. From here {\it Mathematica} is able to numerically find the region in which its corresponding positivity constraints  are satisfied. We increase the number of constraints until a satisfactorily small region is found. 

The manner in which our algorithm differs from \cite{bootstraps} is twofold. First, in the multi-matrix models loop equations are generated for all possible words. One might expect that this would hinder us from finding the search space by introducing more moments. In fact, we found that this was helpful in finding the search space. Secondly, we are working with Dirac ensembles which have both matrix moments and Dirac moments, allowing us to derive more positivity constraints than  if we were working with just a matrix model.

\section{One Matrix Dirac Ensembles}

Consider real finite spectral triples $(A, \mathcal{H}, D)$ where the algebra is $ A= M_N(\mathbb{C})$ and the Hilbert space is $ \mathcal{H} =\mathbb{C} \otimes M_N(\mathbb{C})$. The two signature one noncommutative geometries from \cite{Barrett2016} are
\begin{enumerate}
	\item Type (1, 0) with \begin{equation*}
	\gamma^{1} = 1,
	\end{equation*}
	\begin{equation*}
	D = \{H,\cdot\},
	\end{equation*}
	where $H$ is a Hermitian matrix.
	\item Type (0, 1) with \begin{equation*}
	\gamma^{1} =-i,
	\end{equation*} 
	\begin{equation*}
	D = \gamma^{1}\otimes[L,\cdot],
	\end{equation*}
	where $L$ is a skew-Hermitian matrix.
\end{enumerate}
For  each geometry we define  a quartic action and a partition function
\begin{equation*}
\int_{\mathcal{G}} e^{-\left(g\tr D^{2} + \tr D^{4}\right)}dD.
\end{equation*} 

The trace powers in the action for type $(1, 0)$ can be written in terms of $H$ as
\begin{equation*}
	\tr D^2 = 2N\tr H^2 + 2\tr H \tr H,  
\end{equation*}
and 
\begin{equation*}
	\tr D^4 = 2N \tr H ^4 + 8 \tr H \tr H^3 + 6 \tr H^2 \tr H^2.
\end{equation*}
Similarly for type $(0, 1)$ we have
\begin{equation*}
\tr D^2 = -2N\tr L^2 + 2\tr L \tr L, 
\end{equation*}
and 
\begin{equation*}
\tr D^4 = 2N \tr L ^4 - 8 \tr L \tr L^3 + 6 \tr L^2 \tr L^2.
\end{equation*}

As one can now see, these models are single matrix, multi-trace random matrix models.  They have been found to share many similar properties, such as a genus expansion \cite{Shahab thesis}, with single trace single matrix models. Furthermore, many techniques to analyze them have been extended from single trace models such as the coloumb gas method \cite{First paper} and blobbed topological recursion \cite{AK}. 

For the quartic potential, in the large $N$  limit, when $L$ is replaced with $iH$ for some Hermitian matrix $H$, these are the same matrix and Dirac ensembles \cite{First paper}. Furthermore, it is not hard to see that the trace of $\ell$-th power of the Dirac operators is given by
\begin{equation*}
	 \sum_{k=0}^{\ell} {\ell\choose k} \tr H^{\ell-k} \tr H^{k}.
\end{equation*} 
Hence, by the factorization theorem the Dirac moments in the large $N$ limit are 
\begin{equation*}
	d_{\ell}=\lim_{N \rightarrow \infty} \frac{1}{N^2} \langle \tr D^{\ell}\rangle  =  \sum_{k=0}^{\ell} {\ell\choose k} m_{\ell-k} m_{k},
\end{equation*}
and the general  loop equations  of this model are as follows
\begin{align*}
\sum_{k=0}^{\ell-1} m_k m_{\ell-k-1} = g\left(4 m_{\ell+1} + 4m_1 m_\ell \right) + 8m_{\ell+3}+8m_3m_\ell+24m_1m_{\ell+2} + 24m_2m_{\ell+1}.
\end{align*} 
The odd moments are zero since we are taking the integral of odd functions. This simplifies the loop equations to the following form
\begin{align*}
m_{2\ell+2} = \frac{1}{8} \sum_{k=0}^{2\ell-2} m_k m_{2\ell-k-2} - \frac{1}{2} g m_{2\ell} - 3m_2m_{2\ell}.
\end{align*} 
It is clear from the above recursion that once we have $m_2$, we can find all moments of the model. Hence, the search space has dimension one. Using various positivity constraints on the above loop equations, we were able to approximate $m_2$ with respect to $g$.
\begin{figure}[H]
	\centering
	\includegraphics[width=0.5\textwidth]{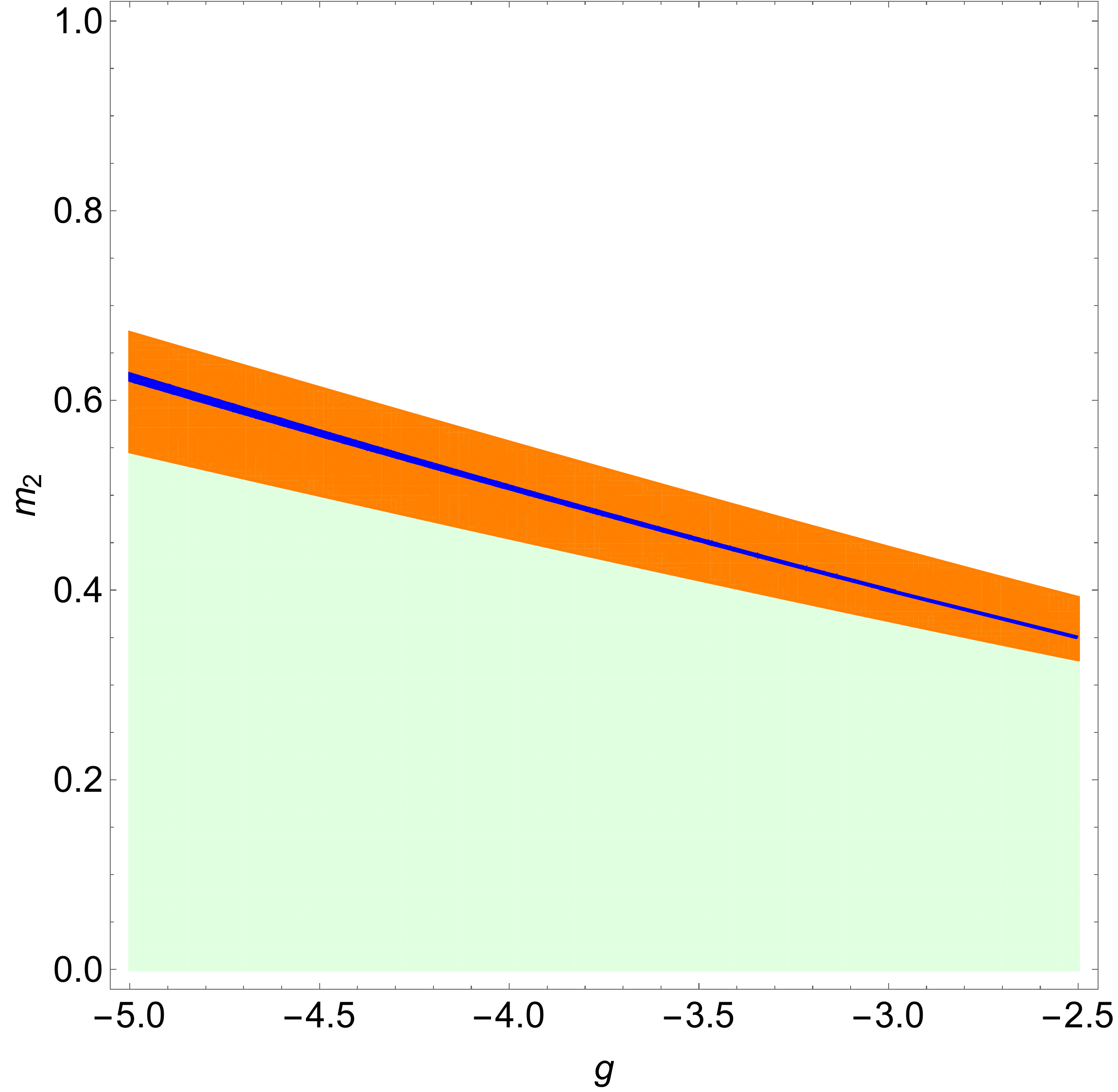}
	\caption{The approximate relation between $m_2$ and $g$, with $g$ varying from $-5$ to $-2.5$.  The different coloured regions denote different constraints applied.}
\end{figure}
\begin{figure}[H]
	\centering
	\includegraphics[width=0.5\textwidth]{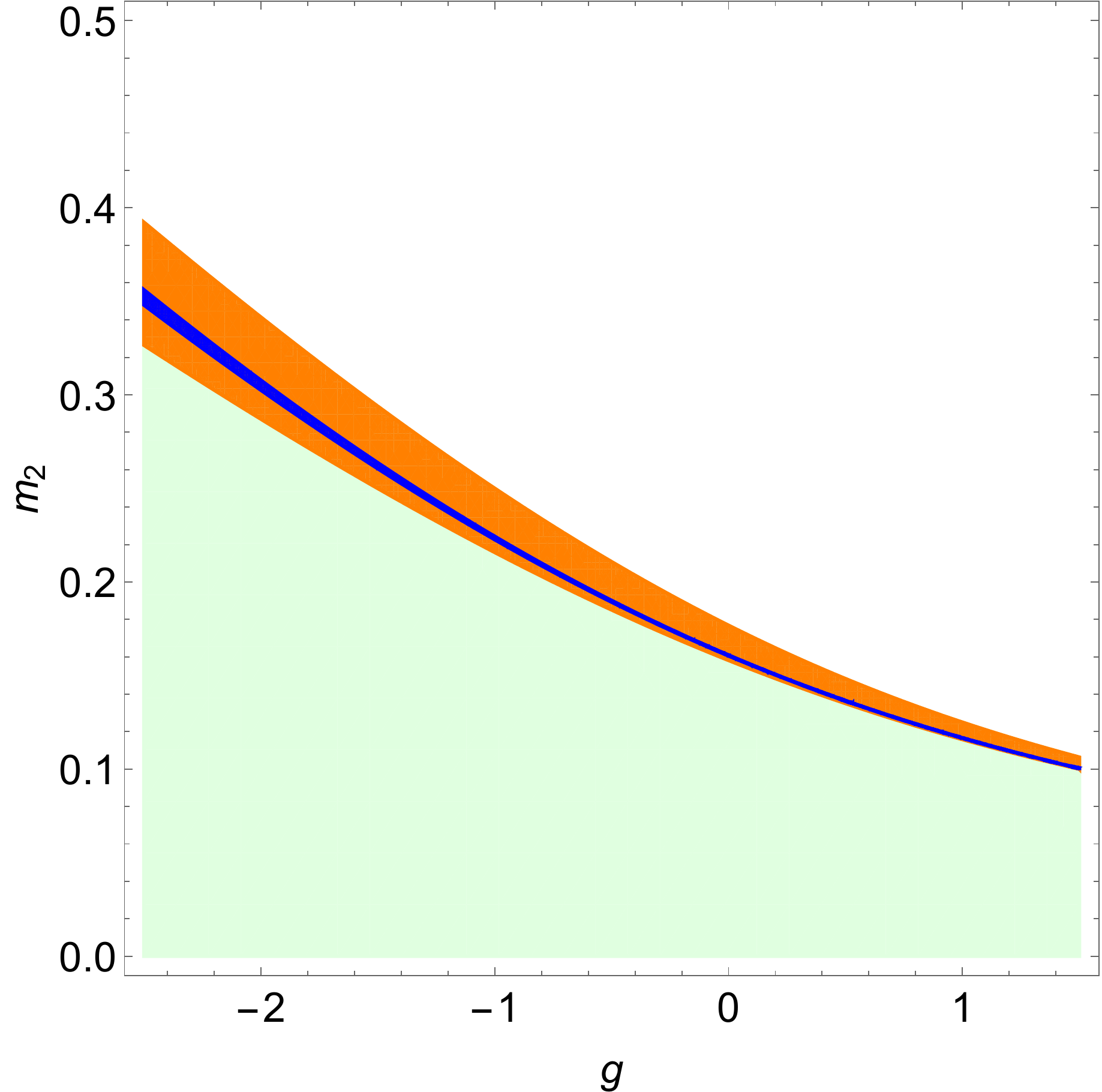}
	\caption{The approximate relation between $m_2$ and $g$, with $g$ varying from $-2.5$ to $1.5$.}
\end{figure}

\pagebreak

The relationship found is remarkably similar (after re-scaling the appropriate factor of a half) to  the analytic relationship found in \cite{First paper}.  In the analytic solution,  for $g$  below the critical value, the relationship is linear, which is clearly visible in the bootstrap solution. For values of $g$ above the critical point, the curve  is also very similar to its analytic counterpart.

\section{Two Matrix Dirac Ensembles}
Consider real finite spectral triples $(A, \mathcal{H}, D)$ where the algebra is  $A= M_N(\mathbb{C})$ and the Hilbert space is $ \mathcal{H} = \mathbb{C}^{2} \otimes M_N(\mathbb{C})$. The three signature two noncommutative geometries from \cite{Barrett2016} are characterized by their Dirac operators $D$ as follows:
\begin{enumerate}
	\item Type (2, 0):  Let  \begin{equation*}
	\gamma^{1} = \begin{pmatrix}
	1 & 0 \\
	0 & -1
	\end{pmatrix}, 
	\quad \quad \gamma^{2} = \begin{pmatrix}
	0 & 1 \\
	1 & 0
	\end{pmatrix}.
	\end{equation*}
	Then, 
	\begin{equation*}
	D = \gamma^{1}\otimes \{H_{1},\cdot\} + \gamma^{2} \otimes \{H_{2},\cdot\},
	\end{equation*}
	where $H_1$ and $H_2$ are Hermitian matrices. 
	
	\item Type (1,1):  Let  \begin{equation*}
	\gamma^{1} = \begin{pmatrix}
	1 & 0 \\
	0 & -1
	\end{pmatrix},
	\quad \quad \gamma^{2} = \begin{pmatrix}
	0 & 1 \\
	-1 & 0
	\end{pmatrix}.
	\end{equation*}
	Then,  
	\begin{equation*}
	D = \gamma^{1}\otimes \{H,\cdot\} + \gamma^{2} \otimes [L,\cdot],
	\end{equation*}
	where $H$ is Hermitian and $L$ is skew-Hermitian.
	
	\item Type (0,2); Let \begin{equation*}
	\gamma^{1} = \begin{pmatrix}
	i & 0 \\
	0 & -i
	\end{pmatrix},
	\quad \quad \gamma^{2} = \begin{pmatrix}
	0 & 1 \\
	-1 & 0
	\end{pmatrix}.
	\end{equation*}
	Then  
	\begin{equation*}
	D = \gamma^{1}\otimes [L_{1},\cdot] + \gamma^{2} \otimes [L_{2},\cdot],
	\end{equation*}
	where $L_1, L_2$ are both skew-Hermitian. 
\end{enumerate}

Now for each geometry  we consider a   quartic action and the partition function
\begin{equation*}
\mathcal{Z}=	\int_{\mathcal{G}} e^{-\left(g\tr D^{2} + \tr D^{4}\right)}dD.
\end{equation*}
 We shall substitute
 $L = iH$, where $H$ is Hermitian, for each skew-hermitian matrix $L$ in the above geometries. The result is that all three matrix models have the same matrix action in the large $N$ limit. To see the explicit potentials we refer the reader to the appendix of \cite{Barrett2016}. Hence, each of these matrix models is identical in this sense, allowing the following results to apply to all three. However, we should not confuse a random matrix ensemble with its Dirac ensemble; their eigenvalues are certainly related but their relationship depends on the geometry. Since the action of this model is even, the odd moments are zero. This means that any moment of a word containing either an odd number of $H_1 $ or $H_2$ is zero. The following terms are the only ones that contribute to the loop equations in the large $N$ limit
 \begin{equation*}
 	\tr D^2 = 4N\tr H_1^2 + 4N\tr H_2^2,
 \end{equation*}
 and 
 \begin{align*}
 	\tr D^4 & =  4N \tr H_1^4 + 4N \tr H_2^4  + 16 N\tr H_1^2H_2^2 -8N \tr H_1H_2H_1H_2 \\
 	&+ 12  (\tr H_1^2)^2  + 12  (\tr H_2^2)^2 + 8 \tr H_1^2 \tr H_2^2.
 \end{align*}

Let us consider the words $H_1^\ell$ for $\ell \geq 1$. The loop equations of this model with respect to these words come from
\begin{align*}
\sum_{i,j=1}^{N}\int_{\mathcal{H}_{N}^{2}} \frac{\partial}{\partial (H_{1})_{ij}} (H_1^\ell)_{ij} \,e^{-\left(g\tr D^{2} + \tr D^{4}\right)}dH_{1}dH_{2} =0,
\end{align*}
giving us
\begin{align}\label{loop equations (2,0)}
	\sum_{k=0}^{\ell-1} m_k m_{\ell-k-1} = (8g+64m_2) m_{\ell+1} + 16m_{\ell+3}-16m_{\ell,1,1,1}  +32 m_{\ell+1 , 2 },
\end{align}
 in the large $N$ limit. Here, we use the notation 
\begin{equation*}
	m_{a,b,c,d} = \lim_{N \rightarrow \infty}\frac{1}{N} \langle \tr H_1^a H_2^bH_1^cH_2^d \rangle. 
\end{equation*}

When $\ell \leq  7$, in the left hand side there is no term that is a product of moments that come from degree four words or higher. For example $m_{2,2}m_{4}$ cannot be found. Thus, loop equations of words with length less than $7$ can be seen as a system of linear equations and that can be solved in terms of $g$ and $m_2$.

\begin{prop}\label{lemma}
	The number of non-trivial moments (up to cyclic permutation and symmetry) that appear in the loop equations of words with length $\ell \geq 9$ is less than the number of non-zero loop equations.
\end{prop}

\begin{proof}
	Denote by $W$ the set of all words with length $\ell+3$. Note that the degree of new moments that appear in equation (\ref{loop equations (2,0)})  is $\ell + 3$.  The set $W$ is acted on by  $A=\mathbb{Z}/{(\ell+3)}\mathbb{Z}$, which shifts letters. Then it follows from Burnside's lemma that
\begin{equation*}
 |W/A|= \frac{1}{\ell +3} \sum_{a\in A} |W^a| \leq \frac{1}{\ell +3} \left( 2^{\ell+3}+\sum_{j=1}^{l+3}    2^{\text{min}(j,l+3-j)}                  \right) \leq  \frac{1}{\ell +3} \left( 2^{\ell+3} + 2^{\frac{\ell+7}{2}}  \right),
\end{equation*}
where  $W^a$ denotes the set of words left invariant by $a\in A$. 
Considering the symmetry property and vanishing of  odd moments, we will get

\begin{equation}\label{moments}
\#\,\text{new moments}	 \leq \frac{1}{\ell +3} \left( 2^{\ell+1} + 2^{\frac{\ell+3}{2}}\right).
\end{equation} 
We may have similar loop equations for two different (up to cyclic symmetry property) words,  but it  is not hard to see that it is never the case that more than half of them are identical. Using symmetry property we have

\begin{equation}\label{SDE's}
 \#\,\text{new loop equations} \geq 2^{\ell-2}  .
\end{equation} 

It follows from equations (\ref{moments}) and (\ref{SDE's}) that for $\ell \geq 9$

\begin{equation*}
	\#\,\text{new loop equations} \geq \#\text{new moments.}
\end{equation*} 

\end{proof}

\begin{corollary}
The dimension of the search space of the above model is $1$.
\end{corollary}

\begin{proof}
	Inductively we can substitute the lower moments into the new loop equations and solve the system of linear equations in terms of $g$ and $m_2$. By proposition \ref{lemma}, the number of new moments is less than the number of nonzero loop equations for a given $g$. Thus the dimension of the search space is $1$.
\end{proof}

By generating all the loop equations for words up to order ten in \textit{Python} and then using \textit{Maple} we found some remarkable formulas for moments up to order eight strictly in terms of $g$ and the second moment $m_{2}$. Here are a selection of them:

\begin{equation*}
m_{4} = -\frac{1}{8}gm_2 + \frac{1}{64},
\end{equation*}
\begin{equation*}
m_{2,2} = -\frac{1}{8}gm_2 -m_{2}^2 + \frac{1}{64},
\end{equation*}
\begin{equation*}
m_{1,1,1,1} = \frac{gm_{2}}{8} + 2 m_{2}^2 - \frac{1}{64},
\end{equation*}
\begin{equation*}
	m_{6} = {\frac {{g}^{2}}{64}}-{\frac {g}{512}}-{\frac {g{m_{2}}^{2}}{8}}+{\frac {
			3\,m_{2}}{64}}-{\frac {5\,{m_{2}}^{3}}{4}},
\end{equation*}
\begin{equation*}
	m_{4,2} = {\frac {{g}^{2}m_{2}}{64}}+{\frac {g{m_{2}}^{2}}{8}}-{\frac {g}{512}}-{\frac {
			{m_{2}}^{3}}{4}}+{\frac {m_{2}}{64}},
\end{equation*}
\begin{equation*}
	m_{3,1,1,1} = -{\frac {{g}^{2}m_{2}}{64}}-{\frac {3\,g{m_{2}}^{2}}{8}}-{\frac {7\,{m_{2}}^{3}}{4
	}}+{\frac {g}{512}}+{\frac {m_{2}}{64}},
\end{equation*}
\begin{equation*}
	m_{2,1,2,1} = {\frac {{g}^{2}m_{2}}{64}}+{\frac {3\,g{m_{2}}^{2}}{8}}-{\frac {g}{512}}+{
		\frac {11\,{m_{2}}^{3}}{4}}-{\frac {m_{2}}{64}},
\end{equation*}

\begin{equation*}
m_{8}=-{\frac {gm_{2}}{64}}+{\frac {{m_{2}}^{4}}{4}}+{\frac {{g}^{2}}{4096}}+{\frac 
		{{m_{2}}^{2}}{256}}+{\frac{3}{4096}}-{\frac {{g}^{3}m_{2}}{512}}+{\frac {3\,{g
			}^{2}{m_{2}}^{2}}{64}}+{\frac {g{m_{2}}^{3}}{2}}.
\end{equation*}


Using only the explicit formulas in terms of $g$ and $m_{2}$  and some associated positivity constraints we were able to find the following regions in the solution space using Mathematica.
\begin{figure}[H]
\centering
\includegraphics[width=0.5\textwidth]{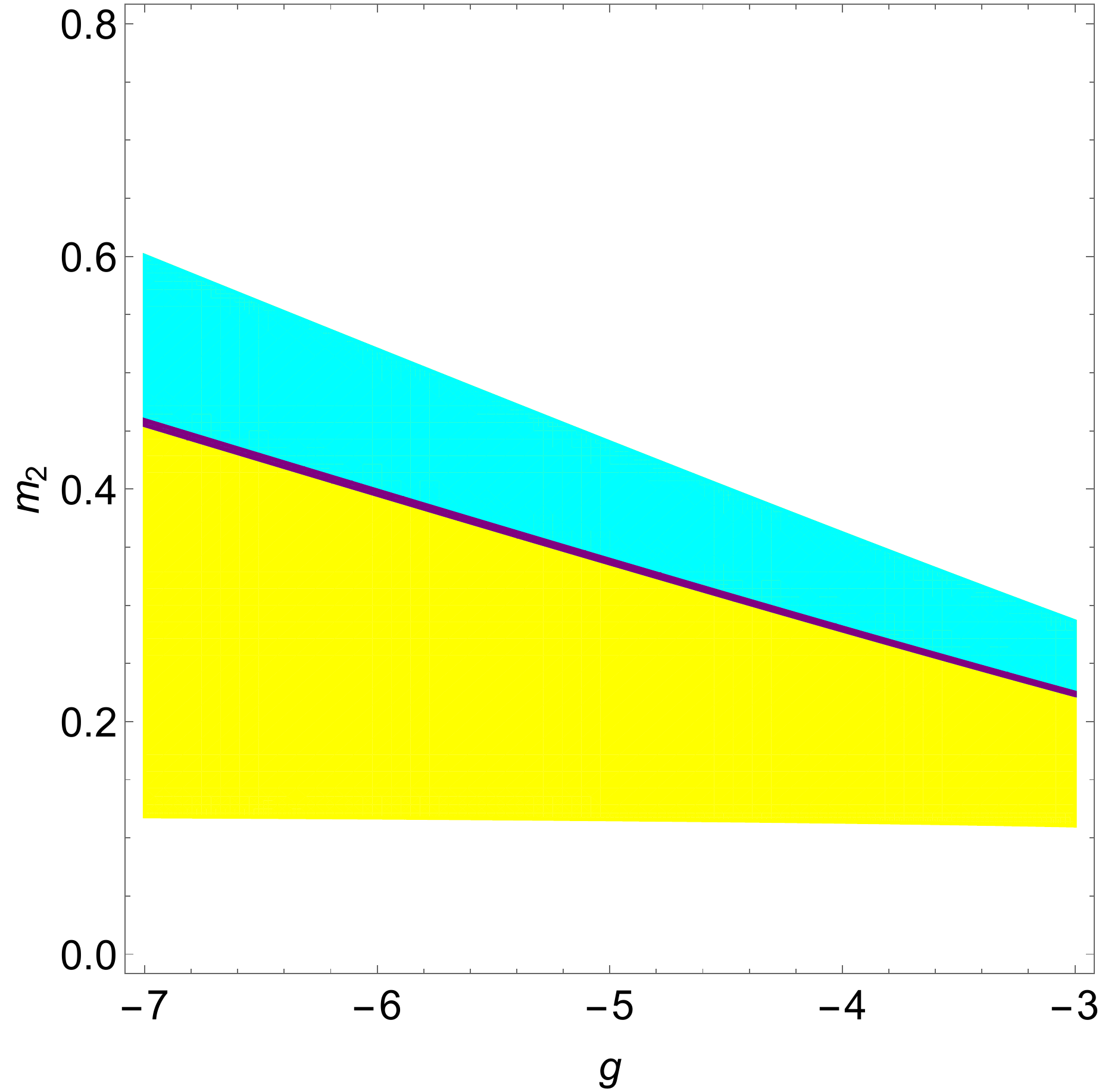}
\caption{The search space region for the (2,0) quartic model where the relationship between $g$ and $m_{2}$ becomes linear.}
\end{figure}

\begin{figure}[H]
	\centering
	\includegraphics[width=0.5\textwidth]{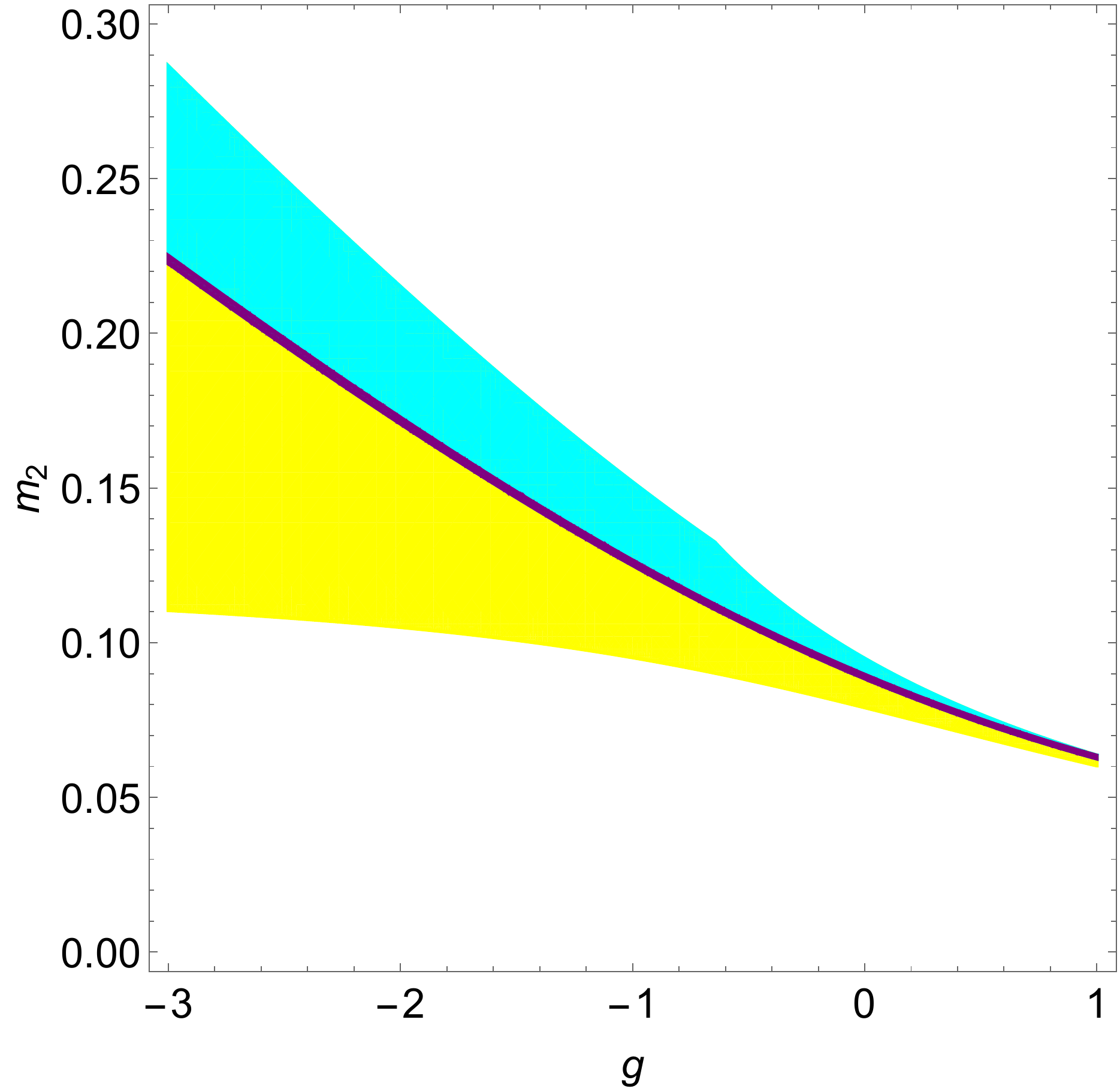}
	\caption{The search space region for the (2,0) quartic model where the relationship between $g$ and $m_{2}$ is  nonlinear.}
\end{figure}
 Note that 
 \begin{equation*}
 	\lim_{N \rightarrow \infty}\frac{1}{N^{2}} \langle \tr D^{2}\rangle = 8 	\lim_{N \rightarrow \infty} \frac{1}{N}\langle \tr H_{1}^{2} \rangle. 
 \end{equation*}
 With this factor of 8 taken into account the figure is remarkably similar to the one computed for ten by ten matrices in \cite{glaser}.  Unlike with the signature one models, no  explicit formula is known for the moments of the Dirac ensemble of signature two models. However, they can be computed with much effort using the combinatorics of chord diagrams; see \cite{Sanchez}. With the above formulas and those from \cite{Sanchez} we have computed the first three nonzero moments of the signature two Dirac ensembles in terms of $m_{2}$ and $g$:
 \begin{equation*}
 	d_{2} = 8\, m_{2},
 \end{equation*}
 \begin{equation*}
 	d_{4} = -4\,gm_2+{\frac{1}{2}},
 \end{equation*}
\begin{equation*}
d_{6} = -160\,{m_2}^{3}-16\,g{m_2}^{2}+6\,m_2+2\,{g}^{2}m_2-{\frac {1}{4}}g.
\end{equation*}
It is also worth noting that the relation between $m_{2}$ and $g$ appears to be linear for values of $g$ roughly below the phase transition  \cite{glaser}. This is precisely what happened for the signature one model analyzed in \cite{First paper}.
 
 While in \cite{bootstraps} the size of the search space is estimated for both single and multi-matrix models, the multi-matrix model from the (2,0) quartic geometry, despite its complexity, had a search space dimension of one! It is now worth noting that our technique differs from Lin's mainly in that when we look for the search space we examine the loop equations generated by all words up to a given order. A smaller search space dimension here is particularly counter intuitive since using more words means introducing more complicated new moments. We believe this is an artifact of this particular model. 

\section{Summary and Outlook}
In this paper Dirac ensembles and their associated matrix models are analyzed using the bootstrap technique. What is found is in very close agreement with simulation results of   \cite{glaser} and analytic treatment  of  \cite{First paper}. What is particularly interesting is that the relationship between the coupling constant and the second moment of the signature two matrix ensemble is linear after the phase transition. This linear relationship was also found analytically for the signature one matrix models in  a similar range of the coupling constant \cite{First paper}. This finding suggests that there may be a deep relationship between the multi-matrix models studied here and the single matrix models.

It is hoped that the computation of moments found here will be used to learn more about these models. Known analytic results do not extend to geometries with signature of two or higher and Monte Carlo simulations are severely limited by matrix size. Hence, bootstrapping offers a new  opportunity to study Dirac and random matrix ensembles suggested by noncommutative geometry.

We hope to apply the bootstrap method  to  more complicated geometries such as (0,3), studied in \cite{glaser}. The methods outlined in this paper should in theory work for any higher signature geometry. It would also be interesting if one could estimate the supports  of the limiting eigenvalue density functions. This would allow one to reconstruct the eigenvalue distributions of both the Dirac and the random matrix ensembles.

Additionally the formulas found for the Dirac moments seem to exhibit some a pattern. If one could  find the loop equations strictly in terms of Dirac moments for any geometry,  this would be an impactful step towards  a better understanding of these models.

In \cite{Barrett2016} it was speculated that the $(2,0)$ model (among others) exhibited manifold-like behaviour near the phase transition. What is meant by this is that the spectrum visually has similarities with the spectrum of the Dirac operator on $S^{2}$. Its spectral density function is of the form $|\lambda|$. Evidence to support this was found in \cite{Spectral estimators}. The work in \cite{truncated} is promising if it can be applied to these models. Another even more recent approach is the utilization of Functional Renormalization Group techniques on these models found in \cite{Sanchez2}.
\section{Appendix on Factorization}
Suppose a formal model is given as described in Section two,  and further assume that the moments of such a model have a genus expansion,  i.e.
\begin{equation*}
	\langle \tr H_{p}^{\ell} \rangle = \sum_{g\geq 0}N^{1-2g}\,\mathcal{T}^{g}_{\ell}, 
\end{equation*}
and
\begin{align*}
	\langle \tr H_{p}^{\ell_{1}}\tr H_{q}^{\ell_{2}}\rangle &= \langle \tr H_{p}^{\ell_{1}}\tr H_{q}^{\ell_{2}}\rangle_{c} + \langle \tr H_{p}^{\ell_{1}}\rangle \langle \tr H_{q}^{\ell_{2}}\rangle\\
	 &= \sum_{n\geq 0} N^{-2g} \mathcal{T}^{g}_{\ell_{1},\ell_{2}} + \left(\sum_{n\geq 0} N^{1-2g}\mathcal{T}_{\ell_{1}}^{g}\right)\left(\sum_{n\geq 0} N^{1-2g}\mathcal{T}_{\ell_{2}}^{g}\right).
\end{align*} 
Here the coefficients are a formal series that counts the number of genus $g$ maps with a boundary of length $\ell$. The subscript $c$ denotes the sum over connected maps with two boundaries of lengths $\ell_{1}$ and $\ell_{2}$. For more details on sums over maps see \cite{Eynard2018}. Thus, taking the large $N$ limit, we obtain
\begin{equation*}
	\lim_{N \rightarrow \infty} \frac{1}{N}	\langle \tr H_{q}^{\ell} \rangle = \mathcal{T}^{0}_{\ell},
\end{equation*}
and 
\begin{equation*}
	\lim_{N \rightarrow \infty} \frac{1}{N^{2}}  \langle \tr H_{p}^{\ell_{1}}\tr H_{q}^{\ell_{2}}\rangle = \mathcal{T}^{0}_{\ell_{1}} \mathcal{T}^{0}_{\ell_{2}} =	\lim_{N \rightarrow \infty} \frac{1}{N^{2}} 	\langle \tr H_{p}^{\ell_{1}} \rangle	\langle \tr H_{q}^{\ell_{2}} \rangle.
\end{equation*}
It is shown  in \cite{Second paper}    that all the models studied in this paper do indeed have a genus expansion, hence the above factorization holds.

\end{document}